\documentclass[12pt]{article} 
\usepackage[margin=1in]{geometry}
\usepackage{times}  
\usepackage{helvet}  
\usepackage{courier}  
\usepackage[hyphens]{url}  
\usepackage{graphicx} 
\usepackage[square,numbers]{natbib} 
\usepackage{caption} 
\usepackage{algorithm}
\usepackage{algorithmic}

\usepackage{newfloat}
\usepackage{listings}
\DeclareCaptionStyle{ruled}{labelfont=normalfont,labelsep=colon,strut=off} 
\floatstyle{ruled}
\newfloat{listing}{tb}{lst}{}
\floatname{listing}{Listing}

\usepackage{amsmath,amsthm,amssymb,cleveref,enumitem,mathtools,pgfplots,tikz}
\usetikzlibrary{shapes.geometric}

\allowdisplaybreaks

\newtheorem{theorem}{Theorem}[section]
\newtheorem{lemma}[theorem]{Lemma}
\newtheorem{proposition}[theorem]{Proposition}
\newtheorem{corollary}[theorem]{Corollary}

\newtheorem{claim}{Claim}

\usepackage{eucal} 
\usepackage{subfigure}
\usepackage{xspace}

\DeclareMathOperator*{\argmax}{arg\,max}
\DeclareSymbolFont{extraup}{U}{zavm}{m}{n}
\DeclareMathSymbol{\vardiamond}{\mathalpha}{extraup}{87}

\usepackage{booktabs}
\usepackage{calrsfs,mathrsfs}

\usepackage{multirow}
\usepackage{pifont}

\usepackage{authblk}

\begin{document}

\title{Envy-free House Allocation under Uncertain Preferences}

\author[1]{Haris Aziz\thanks{haris.aziz@unsw.edu.au}}
\author[1]{Isaiah Iliffe\thanks{i.iliffe@student.unsw.edu.au}}
\author[2]{Bo Li\thanks{bo.li@polyu.edu.hk}}
\author[1]{Angus Ritossa\thanks{a.ritossa@student.unsw.edu.au}}
\author[2]{Ankang Sun\thanks{ankang.sun@polyu.edu.hk}}
\author[1]{Mashbat Suzuki\thanks{mashbat.suzuki@unsw.edu.au}}
\affil[1]{UNSW Sydney}
\affil[2]{Hong Kong Polytechnic University}
\date{\vspace*{-1cm}}

\maketitle

\begin{abstract}
Envy-freeness is one of the most important fairness concerns when allocating resources.   
We study the envy-free house allocation problem when agents have uncertain preferences over items and consider several well-studied preference uncertainty models. 
The central problem that we focus on is computing an allocation that has the highest probability of being envy-free. 
We show that each model leads to a distinct set of algorithmic and complexity results, including detailed results on (in-)approximability. 
En route, we consider two related problems of checking whether there exists an allocation that is possibly or necessarily envy-free.  
We give a complete picture of the computational complexity of these two problems for all the uncertainty models we consider. 
\end{abstract}

\section{Introduction}

Multi-agent resource allocation is one of the fundamental issues at the intersection of computer science and economics. 
We consider a fundamental allocation problem in which agents have preferences over items or houses and each agent is allocated one house while taking into account their preferences. The problem has been referred to as the \textit{house allocation} or \textit{assignment problem}. When agents have deterministic preferences over items, the problem is very well-understood. For example, there are characterizations of Pareto optimal allocations (see, e.g.~\citep{AbSo98a}) and polynomial-time algorithms for computing envy-free allocations (see, e.g. \citep{GSV19}).

{In this paper, we focus on fairly allocating items in the house allocation problem. }
The central fairness concept we consider is envy-freeness (EF) which is considered one of the gold-standard for capturing fairness. 
{An allocation is EF if every agent gets her favorite house among all the houses assigned to the agents. }
While the structure of EF allocations under deterministic preferences is well-understood, there is little prior work on the complexity of computing EF house allocations under uncertain preferences. Uncertain preferences are important to model when agents’ preferences may not be completely known, and the central planner may have to make decisions based on limited information. The uncertainty model can also capture the situation when the agents represent a group of agents with a composition of preferences. 
{Various types of uncertain preferences have been examined in the literature (see, e.g, \citep{HAK+12a,ABH+19,ABG+20a,ABF+22}), including the lottery Model, the joint probability model, and the compact indifference model.}
{All the models are well motivated and some real-world applications can be found in, for example, \citep{ABG+20a}.}


Given an instance in which agents' uncertain preferences are given as input, the central problem that we consider is  {\sc Max-ProbEF} which concerns computing an allocation with the highest probability of being EF. Such an allocation can be viewed as being \textit{robustly} fair under uncertain information. 
En route, we consider two related problems: {\sc ExistsPossiblyEF} (i.e., does there exist an allocation that is EF with non-zero probability?) and  {\sc ExistsCertainly\-EF} (i.e., does there exist an allocation that is EF with probability one?).
Our work aims to present a comprehensive study of the computation complexity for all these problems under various uncertainty models.

      \begin{table}[h]
          \centering
          \scalebox{0.75}{
          \begin{tabular}{llllll}
              \toprule
           &Lottery&Compact&Joint\\
              &&indifference&Prob\\
Problems&&&&&\\
               \midrule
     {\sc ExistsPossiblyEF}& NP-c &in P&in P\\
  \midrule
      {\sc ExistsCertainlyEF}&NP-c&in P&NP-c\\
      \midrule

       \multirow{2}{*}{\sc Max-ProbEF}&NP-h&NP-h&NP-h\\
&$\clubsuit$ &$\spadesuit$ $\dagger$& $\vardiamond$ \\ 
              \bottomrule
          \end{tabular}
          }

          \caption{Summary of the main results. The symbol $\clubsuit$ indicates that the problem admits no bounded multiplicative approximation, assuming P$\neq$NP.
The symbol $\spadesuit$ indicates that the problem admits no polynomially-bounded multiplicative approximation, assuming P$\neq$NP.
The symbol $\vardiamond$ indicates that there is no polynomial-time algorithm with better than $\frac{6}{5}$-approximation ratio, assuming P$\neq$NP.
The symbol $\dagger$ indicates the problem admits a polynomial-time algorithm that either solves \textsc{Max-ProbEF} exactly or certifies that every allocation has a small probability of being EF.  
} 
          \label{table:summary:uncertainEF}
      \end{table}

\subsection*{Contributions}

In this work, we consider three natural and well-motivated uncertain preference models, namely, the lottery model, the compact indifference model, and the joint probability model.
In the lottery model, each agent has an independent probability distribution over linear preferences. 
In the joint probability model, the input is a probability distribution over preference profiles.
The compact indifference model is a special type of lottery model in which the input is each agent having a corresponding weak order of the items and each complete linear order extension of this weak order is assumed to be equally likely.   
For each model, we first present the computational complexity results of checking whether there exists an allocation that is possibly envy-free ({\sc ExistsPossiblyEF}) or necessarily envy-free ({\sc ExistsCertainly\-EF}). 
Then we undertake a detailed complexity analysis of the central problem  {\sc Max-ProbEF} of computing an allocation that has the highest probability of being envy-free.
Our results are summarized in Table~\ref{table:summary:uncertainEF}. 


\paragraph{Lottery Model}
We start with the lottery model. We prove that both {\sc ExistsPossiblyEF} and {\sc ExistsCertainly\-EF} are NP-complete. 
The intractability of {\sc ExistsPossiblyEF} also implies that there is no polynomial-time algorithm with bounded multiplicative approximation ratio for {\sc Max-ProbEF}, assuming P$\neq$NP. 

\paragraph{Compact Indifference Model}
 In sharp contrast to the lottery model, we show that there exist polynomial-time algorithms for {\sc ExistsPossiblyEF} and {\sc ExistsCertainly\-EF}. 
However, the main problem {\sc Max-ProbEF}  continues to be intractable: we actually prove that finding an $f(n, m)$-approximation of {\sc Max-ProbEF} under the compact indifference model is NP-hard where $f(n, m)$ is an arbitrary polynomial in the numbers of agents $n$ and houses $m$. 
We complement this result by presenting a central algorithmic result: for any fixed $\epsilon >0$, there exists a polynomial-time algorithm that either computes the optimal solution for {\sc Max-ProbEF} exactly, or returns a certificate that every allocation has EF probability less than $\epsilon$.  
We regard this part as the main contribution of this work.


\paragraph{Joint Probability Model}
Finally, we study the joint probability model. 
We first show that {\sc ExistsPossiblyEF} can be solved in polynomial time, while {\sc ExistsCertainly\-EF} is NP-complete. 
Then we prove that there is no polynomial-time algorithm with a $(\frac{6}{5} - \epsilon)$-approximation ratio for {\sc Max-ProbEF} for any constant $\epsilon>0$,  assuming P$\neq$NP.

\section{Related Work}

Our work combines aspects of envy-free allocation and uncertainty in preferences that have been explored in the context of stable marriage, voting, and Pareto optimal allocation.


The house allocation problem is one of the most fundamental and well-studied problems in economics and computer science ~\citep{AbSo99a,ACMM05a,ABL+16a,BoMo01a,Gard73b,Sven94a,Sven99a}.  Typically, the problem has an equal number of agents and houses and the goal is to allocate one house to each agent. In our setup, we allow the number of houses to be more than the number of agents.
\citet{GSV19} presented an elegant polynomial-time algorithm to check whether an envy-free allocation exists or not in which each agent gets one item. 
\citet{AS22} also presented similar arguments for envy-free outcomes for bipartite graphs. 


Uncertainty in preferences has already been studied in voting~\citep{HAK+12a}.  They examine the computation of the probability of a particular candidate winning an election under uncertain preferences for various voting rules such as Plurality, $k$-approval, Borda, Copeland, and Bucklin etc.
\citet{ABH+19} explore the Pareto optimal allocation under uncertain preferences. We consider the same problem setup and preference uncertainty
models as them but instead of focusing on Pareto optimality, our central property is envy-freeness.
\citet{ABG+20a,ABF+22} examined uncertain preferences in the context of two-sided matching. The central property they examine is stability and they compute matchings that have the highest probability of being stable. 

%
%

{In a related line of work, \citet{DGKPS14} initiated the study of existence of envy-free allocations when agent valuations for items are drawn from a probability distribution. 
In follow up work, \citet{MS19,MS21} further refined the parameter ranges for which envy-free allocations exist with high probability. There has been several other  works \citep{FGHLPSSY19,BFGP22,BG22} on studying fair division under distributional models. In these models, where agent values are drawn from a distribution, allocations can change after the realization of the coin toss, whereas in our setting, we study a fixed allocation that is \textit{{``robust"}} under uncertainty.

\section{Preliminaries}

An instance of the \emph{(deterministic) house allocation  problem} is a triple $(N,H,\succ)$ where $N=\{1,\ldots, n\}$ is the set of $n$ agents, $H=\{h_1,\ldots, h_m\}$ is the set of $m$ items (also referred to as houses), and the preference profile ${\succ}=(\succ_1,\ldots,\succ_n)$ specifies complete and asymmetric
    	preferences $\succ_i$ of each agent $i$ over $H$. 
    	Note that in the classical allocation problem, agents' preferences are also assumed to be transitive, hence resulting in linearly ordered preferences.
In some examples, we will represent the preferences as an ordered list in decreasing order of preferences from left to right.
Let $\mathcal{R}(H)$ denote the set of all complete and asymmetric relations over the set of items $H$. 
	
    	An \emph{allocation} is an assignment of items to agents such that each agent is allocated a unique item, and each item is allocated to at most one agent. 
Throughout the paper we assume $m\geq n$ as the only envy-free allocation in the case of $m<n$ is one in which no agent gets any item.
For a given allocation $\omega$, let $\omega(i)$ denote the item allocated to agent $i$.
We denote the set of all allocations by $\mathcal{A}$.
	An allocation $\omega$ is \emph{envy-free (EF)} if $\omega(i)\succ_i  \omega(j)$ for $i\in N$ and $j\in N\setminus \{i\}.$

	In this work, we allow agents to express uncertainty in their preferences
    	and consider various uncertainty models. 

 \begin{itemize}
 \item \textbf{Lottery Model}: For each agent $i\in N$, we are given a probability distribution $\Delta_i(\mathcal{R}(H))$ over linear preferences. We assume that the probability distributions are independent. 
 \item \textbf{Compact Indifference Model}: Each agent reports a single weak preference list that allows for ties. Each complete linear order extension of this weak order is assumed to be equally likely. We use $a\succsim_i b$ to represent that $a$ is weakly preferred by $i$ over $b$. We use $a\sim_i b$ to represent that agent $i$ is indifferent between $a$ and $b$ in the weak preference order. 
%
\footnote{The assignment problem is also known as the House Allocation problem. The compact indifference model can be viewed as the assignment problem with ties, or as it is widely known in the literature as the House Allocation problem with Ties (HRT)~\citep{Manl13a},
where any preference list obtained by breaking ties arbitrarily is possible, and all possible preferences have the same likelihood of being realized.}

 \item  \textbf{Joint Probability Model}: A probability distribution $\Delta(\mathcal{R}(H)^n)$  over linear preference profiles is specified where a preference profile specifies (deterministic) preferences of each agent over items.

 \end{itemize}

An uncertain preference model is \emph{independent} if any uncertain preference profile $L$ under the model can be written as a product of uncertain preferences $L_i$ for all agents $i$, where all $L_i$'s are independent. Note that the joint probability model is not independent in general, but all the other problems that we study are independent.

For the uncertainty models, we consider the following corresponding problems:

\begin{itemize}
    \item {\sc Max-ProbEF}: Given an instance of the problem, compute an allocation that maximizes the probability of being envy-free (EF). Formally,
$$\argmax\limits_{w\in \mathcal{A}} \Pr\limits_{D\sim \Delta(\mathcal{R}(H)^n)}[w \text{ is EF under profile }D ].$$
\item {\sc ExistsCertainly\-EF}: Determine whether there exists an allocation that is EF with probability one.
\item {\sc ExistsPossibly\-EF}: Determine whether there exists an allocation that is EF with non-zero probability.
\end{itemize}
Note the answer to {\sc Max-ProbEF}  also gives an answer to {\sc ExistsCertainly\-EF} and {\sc ExistsPossibly\-EF}.

\section{Initial Structural \& Algorithmic Results}

In this section, we present some initial structural and algorithmic results. First we show that given an allocation, the probability that it is EF can be computed efficiently. 

\begin{proposition}\label{prop:efprob}
Given allocation $w$, the probability that $w$ is EF can be computed in polynomial time for the (i) joint probability model, (ii) lottery model, and (iii) compact indifference model.
\end{proposition}

The argument for the joint probability model is trivial. For the other models, for each agent $i\in N$, we find the probability $q_i$ that the agent is not envious. The probability that $w$ is envy-free is equal to the probability that all agents are envy-free which is computed as $\prod_{i\in N}q_i$. The details are in the appendix.

Next, we present some structural results that suggest that the main challenge of $\operatorname{Max-ProbEF}$ lies in determining which houses are included in the matching. 
We say that an uncertain preference model is \textit{reasonable} if for any  set $M'\subset M$ with $|M'|=n$ of houses, and any $i\in N$ and $j\in M'$, 
the probability $p_{ij}$ that $j$ is the most preferred house for agent $i$ among houses in $M'$ can be computed in polynomial time.  Note that all the models we consider are reasonable. For example, we argue why the lottery model is reasonable. Let $\Delta_i=(\lambda_{r},\succ^r_i)_{r\in S }$ be the distribution where agent $i$'s preference is drawn i.e., preference $\succ^r_i$ is drawn with probability $\lambda_{r}$ from the set of preferences $\{\succ^r_i\}_{r\in S}$ with positive support. Note that $p_{ij}$ can be computed efficiently since
$p_{ij}=\sum\limits_{\{r\in S \ | j \succ^r_i \ell, \forall \ell\in M'\setminus j \}}\lambda_{r}.$  



%

\begin{proposition}\label{prop: Restrict}
For any reasonable uncertainty model that is independent, given a set $M'\subset M$ with $|M'|=n$ of houses, an allocation of $M'$ to $N$ which maximizes the probability of EF can be computed in polynomial time. 
\end{proposition}
\begin{proof}
Construct a complete weighted bipartite graph $G=(N \cup M',E)$ with edge weights $w_{ij}=\log(p_{ij})$ for each $i\in N$, $j\in M'$. Here $p_{ij}$ denotes the probability that the house $j$ is agent $i$'s favourite when the houses are restricted to $M'$. 
Observe now that the maximum weight matching in $G$ gives an allocation that maximizes the probability of EF when the houses are restricted to $M'$. This holds since $p_{ij}$ is the probability that $i$ does not envy any other agent when assigned $j$, and thus the probability that an allocation $w$ is envy free is $\prod_{i\in N} p_{i,w(i)}$. Furthermore, the allocation that maximizes  $\prod_{i\in N} p_{i,w(i)}$ also maximizes $\sum_{i\in N}\log(p_{i,w(i)})$. 
\end{proof}

Proposition~\ref{prop: Restrict} implies that in order to find the optimal solution to $\operatorname{Max-ProbEF}$ it suffices to find the houses that are assigned in the optimal solution rather than the allocation itself. Hence we may focus our attention on finding a set of houses $M^*$ with $|M^*|=n$ such that $G=(N \cup M^*,E)$ has the highest max weight matching.
Our insights give the following algorithmic result. 



\begin{proposition}\label{prop: EF n+k}
  For any independent reasonable uncertainty model, $\operatorname{Max-ProbEF}$ problem with $m=n+k$ houses can be solved in polynomial time, for any constant $k\in \mathbb{N}$.  
\end{proposition}
\begin{proof}
We iterate through each of the ${m \choose n}$ house restrictions and run the polynomial-time algorithm proposed in Proposition~\ref{prop: Restrict}. Return the allocation with the highest max weight matching among the ${m \choose n}$ bipartite graphs. This procedure is polynomial-time since the algorithm in Proposition~\ref{prop: Restrict} is called at most $ {m \choose n} = {n+k \choose n}= {n+k \choose k}=O(n^k)$ times. 
Note that aforementioned procedure outputs an optimal solution since it outputs an  allocation that maximizes the probability of EF under  each possible house restriction.
\end{proof}

%



 \section{Lottery Model}\label{sec: Lottery}

In this section, we study the lottery model.

 \begin{theorem}\label{theorem:Lottery-CertainlyEF}
 	For the lottery model, {\sc{ExistsCertainlyEF}}
 	is NP-complete.
 \end{theorem}

\begin{proof}[Proof of Theorem~\ref{theorem:Lottery-CertainlyEF}]
	The problem {\sc{ExistsCertainlyEF}}
	is in NP because it can be checked in polynomial time whether a given allocation is certainly EF or not.
	
	To prove NP-hardness, we reduce from the problem of Restricted 3-Exact Cover:
	given a family $F = \{ S_1,\ldots, S_n\}$
	of $n$ subsets of $S = \{ u_1, \ldots, u _ { 3 m }\}$,
	each subset of $F$ has a cardinality three, and moreover each element in $S$ appears in exactly three subsets of $F$.
	Is there a subfamily of $ m $ subsets that covers $S$?
	
	We now construct a house allocation instance under lottery model with
	a set $\{1,\ldots, 3m\}$ of $3m$ agents and a set $H = \{ h ^ 1 _ 1, h^2_1, h^3_1, h_2^1, h_2^2, h_2^3, \ldots, h_n^1, h_n^2, h_n^3 \}$ of $ 3n $ houses.
	Intuitively, every agent $ i $ corresponds to the element $ u _ i $ of Restricted 3-Exact Cover problem.
	For every subset $S_j $,
	we construct three houses $ h^1_j, h^2_j, h^3_j$.
	Then, we construct the preferences lists of agents.
	Fix agent $ i $ and let $ S _ {i_1}, S _ { i _ 2}, S_{i_3}$
	be the three subsets including element $ u _ i $ in the Restricted 3-Exact Cover instance.
	Then, for each $ S _ { i _ j }$,
	we construct two preference lists for agent $ i $, and hence in total six preference lists for each agent.
	In every linear preference, agent $ i $ prefers the nine houses corresponding to the
	three subsets containing element $ u _ i $
	than other houses of which the order can be arbitrary.
	For convenience, we only provide partial preference lists.

	For $j=1,2,3$, suppose that element $ u _ i $ is the one with $l(j)$-th smallest index in $S_{ i _ j }$, and moreover, let $ \{ p(j) , q(j) \} = \{ 1,2,3\} \setminus \{ l(j) \}$ be the set of other two indices.
	Agent $i$'s two preferences corresponding to $ S _ { i _ 1 }$ are
	
	\begin{enumerate}
		\item $h _{ i _ 1 } ^ {l(1)} \succ h _{ i _ 1 } ^ {p(1)}\succ h _{ i _ 1 } ^ {q(1)} \succ h _ { i _ 2 } ^ {l(2)}\succ h _ { i _ 2}^ {p(2)} \succ h _ { i _ 2}^ {q(2)} \succ h _ { i _ {3} } ^ {l(3)}\succ h _ { i _ {3}}^ {p(3)} \succ h _ { i _ {3}}^ {q(3)} \succ \textnormal{other houses}$;
		\item $h _{ i _ 1 } ^ {l(1)} \succ h _{ i _ 1 } ^ {q(1)}\succ h _{ i _ 1 } ^ {p(1)} \succ h _ { i _ 2 } ^ {l(2)}\succ h _ { i _ 2}^ {p(2)} \succ h _ { i _ 2}^ {q(2)} \succ h _ { i _ {3} } ^ {l(3)}\succ h _ { i _ {3}}^ {p(3)} \succ h _ { i _ {3}}^ {q(3)} \succ \textnormal{other houses}$;
	\end{enumerate}
	In particular, agent $ i $ prefers the three houses regarding $S _ { i_ 1 }$
	than other houses. 
	Then, among the three houses regarding $ S _ {i _ d }$ ($ d\neq 1 $),
	she likes the one corresponding to her the most.
	For the above two preferences, the only difference is that the second and the third most preferred houses are swapped.
	Similarly, for the two preferences regarding $S _ {i_2}$, 
	agent $ i $ prefers $\{h_{i_2}^j\}$ than other houses,
	and 
	of these she likes $ h _ {i_2}^{l(2)}$ the most.
	For houses regarding $S _ {i_d}$ with $d\neq 2$,
	she always likes the one corresponding to her the most.
	The same idea applies to the two preferences lists regarding $S _ {i_3}$.

	We now prove that we have a ``yes'' Restricted 3-Exact Cover instance if and only if
	there is a certainly EF allocation.
	Without loss of generality,
	denote by $ S_ 1, S_2, \ldots, S_m$ a solution of Restricted 3-Exact Cover instance.
	For each agent $ i $,
	if element $ u _ i $ is included in $ S_  j $
	and $ u _ i $ is the one with $l$-th smallest index,
	then assign house $ h _ j ^ l $ to agent $ i $.
	Since $ S_1,S_2, \ldots, S_m $ is a solution,
	each element appears exactly once in $\bigcup_{ t \in [m]} S _ t $,
	which implies that each agent $ i $ receives one house.
	Let $\omega$ be such an allocation.
	For a contradiction, suppose agent $ i $ is not certainly EF in $\omega$,
	then there must be a preference $\succ$ of agent $i $ together with
	an allocated house
	$ h _ a^b$ such that $ h _ a^b \succ \omega(i)$.
	By agent $ i $'s preferences,
	it holds that $ \omega(i) \notin \{ h_ a ^ 1, h_a^2, h_a^3 \}$ but $ u _ i \in S_ a$.
	Note that
	$ \omega ( i ) \in \{ h _ c ^ 1, h_c^2, h_c^3 \}$ for some $c \neq a$,
	which implies that
	in the ``yes'' solution, there are two different subsets $S_c, S_a$ containing $ u _ i $, a contradiction.
	
	For the other direction,
	suppose that $\omega^*$ is a certainly EF allocation.
	Fix $ i $, recall that element $ u _ i $ is in subsets
	$ S _{ i _ 1},S _ { i  _2}, S _{  i _ 3}$.
	Then, we show that $\omega^*(i)$ must be a house constructed from
	$ S _{ i _ 1}, S _ {i_2}, S _ { i _ 3 } $.
	
	\begin{claim}\label{Claim::certainly-lottery-1}
		$\omega ^ * ( i ) \in \bigcup _ { j =1} ^ t \bigcup _ { k = 1} ^ 3 h _ {i _ j} ^ k .$
	\end{claim}
	\begin{proof}[Proof of Claim~\ref{Claim::certainly-lottery-1}]
		For a contradiction, assume $ \omega^*(i) = h _ p ^ k$ for some $ p \neq i_1,i_2, i_3$.
		By preferences lists,
		there exists another agent $ s $,
		of whom $h_p^k$ is the most prefered house in one preference of him.
		Thus, agent $s$ is not certainly EF, a contradiction.
	\end{proof}
	Next, we show that
	if for some $ k \in \{1,2,3\}$, house $ h _ j ^ k $ is assigned,
	then other two houses constructed from subset $ S_  j $ should also be assigned in $\omega^*$.
	
	\begin{claim}\label{Claim::certainly-lottery-2}
		In allocation $\omega ^*$,
		if for some $ k \in \{ 1,2,3\}$, house $ h _ j ^ k $ is assigned,
		then houses $ h_ j ^1, h_j^2, h_j^3 $ are allocated.
	\end{claim}
	\begin{proof}[Proof of Claim~\ref{Claim::certainly-lottery-2}]
		Without loss of generality, assume $ h _ j ^ 1$ is allocated.
		According to constructed preferences,
		there exists an agent $ p_2 $ (resp., $ p _ 3$) of whom
		house $h _ j ^ 2$ (resp., $h _ j ^ 3$) is the most preferred house in one preference. 
		Then, for agent $ p _ 2$,
		she has a preference in which $ h _ j ^1$ is the second most preferred house, while in the same preferece,
		$ h_ j ^2$ is the house she likes the most.
		Then, house $ h _ j ^2$ must be allocated in $\omega^*$, and moreover,
		can only be assigned to agent $ p _ 2$.
		By similar arguments,
		house $ h _ j ^ 3$ should also be assigned in $\omega^*$.
	\end{proof}
	Note that in total $3m$ houses are allocated in $\omega^*$.
	According to Claim~\ref{Claim::certainly-lottery-2},
	the $3m$ allocated houses must corresponds to $m$ subsets of Restricted 3-Exact Cover instance.
	In other word,
	the $3m$ allocated houses must be
	in the form of
	$ h^1_{l_1}, h^2_{l_1}, h^3_{l_1},
	\ldots, h^1_{l_m},h^2_{l_m}, h^3_{l_m}$.
	Therefore, allocation $\omega^*$ offers us $m$ subsets $S_{l_1}, S_{l_2}, \ldots, S_{l_m}$ for Restricted 3-Exact Cover instance.
	If these $m$ subsets are not a solution,
	then at least one element $ u _ { i ^ {\prime}}$
	is not coverd.
	Then, $\omega^*(i^{\prime})$
	is a house regarding a subset $ S _ { j ^{\prime}} $ with $ u  _{ i ^ {\prime}} \notin S _ { j ^{\prime}} $,
	contradicting Claim~\ref{Claim::certainly-lottery-1}.
	Thus,
	the $m$ subsets $S_{l_1}, S_{l_2}, \ldots, S_{l_m}$ drawn from $\omega^*$ provides a solution of
	Restricted 3-Exact Cover.
\end{proof}

 \begin{theorem}\label{theorem:Lottery-PossiblyEF}
 	For the lottery model, {\sc{ExistsPossiblyEF}}
 	is NP-complete.
 \end{theorem}
 
 We prove Theorem~\ref{theorem:Lottery-PossiblyEF} via a sequence of reductions starting from {\sc{Minimum Coverage}}, which is known to be NP-hard \cite{vinterbo2002}.
For this purpose, we introduce two new problems:
 \begin{itemize}
 \item  {\sc{ExistsPartialEF}} under binary preferences. In this problem, there is a set $[n]$ of agents and a set $H$ of $m$ houses. Each agent has deterministic binary preferences over the houses: in particular, each agent $i$ partitions the houses into two sets $A_i, B_i$ where houses in the same set are valued equally, and $h_a \succ_i h_b$ for all $h_a \in A_i$ and $h_b \in B_i$. 
 
 Additionally, an integer $k$ is supplied (with $k \leq n$). 
 The goal is to determine whether there exists an allocation of houses to $k$ agents such that these $k$ agents are envy-free. The $n-k$ agents without a house do not experience envy. 
 
 \item {\sc{ExistsPartialPossiblyEF}} under the lottery model. This problem is similar to {\sc{ExistsPossiblyEF}}, however an additional integer $k$ is supplied (with $k \leq n$). The goal is to determine whether there exists an allocation of houses to $k$ agents such the probability of envy-freeness is nonzero. The $n-k$ agents without a house do not experience envy. 
 \end{itemize}
 
We show that {\sc{Minimum Coverage}} reduces to {\sc{ExistsPartialEF}}, which in turn reduces to {\sc{ExistsPartialPossiblyEF}} which finally reduces to {\sc{ExistsPossiblyEF}}. To prove the next lemma, we use a modification of the proof from Theorem 3.5 of \citet{KMW21}. They prove hardness for a similar problem, where all $n$ agents must be allocated houses with the requirement that at least $k$ agents are envy-free. 
This is detailed in the appendix.

\begin{lemma}\label{lem:ExistsPartialEF}
	With binary preferences, {\sc{ExistsPartialEF}} is NP-hard.
\end{lemma}

For the next lemma, we reduce from {\sc{ExistsPartialEF}} under binary preferences, which is NP-hard from Lemma~\ref{lem:ExistsPartialEF}.
\begin{lemma}\label{lem:ExistsPartialPossiblyEF}
	For the lottery model, {\sc{ExistsPartialPossiblyEF}} is NP-hard.
\end{lemma}

\noindent
We are now ready to prove Theorem~\ref{theorem:Lottery-PossiblyEF}. 
\begin{proof}[Proof of Theorem~\ref{theorem:Lottery-PossiblyEF}]
	The problem  {\sc{ExistsPossiblyEF}} is in NP because it can be checked in polynomial time whether an allocation is possibly EF or not.
	
	To prove NP-hardness, we reduce from {\sc{ExistsPartialPossiblyEF}} under the lottery model, which is NP-hard from Lemma~\ref{lem:ExistsPartialPossiblyEF}.
	Consider an instance $I$ of {\sc{ExistsPartialPossiblyEF}} under the lottery model with $n$ agents and $m$ houses $\{h_1, \ldots, h_m\}$ and parameter $k$.
	We construct an instance $I'$ of {\sc{ExistsPossiblyEF}} under the lottery model with $n$ agents and $m+n-k$ houses, where the agents in $I'$ correspond to agents in $I$.
	The houses are $\{h_1, \ldots, h_m\} \cup \{e_1, \ldots, e_{n-k}\}$.
	
	Consider some agent $i$. 
	Assume that this agent has $\ell$ preference lists in $I$, and that the $j$-th such preference list is $a_{j, 1} \succ \ldots \succ a_{j, m}$.
	In the instance $I'$, agent $i$ has $\ell+n-k$ preference lists, each with equal probability:
	\begin{itemize}
		\item For each $j \in [\ell]$, we have the list $a_{j, 1} \succ \ldots \succ a_{j, m} \succ e_1 \succ \ldots \succ e_{n-k}$.
		\item For each $j \in [n-k]$, we have the list $e_j \succ e_1 \succ \ldots \succ e_{j-1} \succ e_{j+1} \succ \ldots \succ e_{n-k} \succ a_{1, 1} \succ \ldots \succ a_{1, m}$. Note that the $e_1 \succ \ldots \succ e_{j-1}$ segment is empty if $j = 1$, and the $e_{j+1} \succ \ldots \succ e_{n-k}$ segment is empty if $j = n-k$.
	\end{itemize}
	
	We now prove the answer to $I$ is ``yes'' if and only if the answer to $I'$ is ``yes''.
	Assume the answer to $I$ is ``yes'' and we have an allocation where agents $u_1, \ldots, u_k$ are each allocated house $\omega(u_i)$ in a possibly envy-free way. 
	Let $u_{k+1}, \ldots, u_n$ be the agents that were not allocated a house. 
	
	We create an allocation $\omega'$ for $I'$ as follows:
	\begin{itemize}
		\item For each $i \in [k]$, $\omega'(u_i) = \omega(u_i)$,
		\item For each $i \in [k+1, n]$, $\omega'(u_i) = e_{i-k}$.
	\end{itemize}
	Consider some agent $i \in [k]$, where $\omega'(u_i) = \omega(u_i)$. Since agent $i$ is possibly envy-free in $I$, there exists some preference list in $I$ where $\omega(u_i)$ is preferred over every other allocated house.
	However, this corresponds to a preference list in $I'$ where $\omega'(u_i)$ is preferred to all the other houses allocated in $\omega'$.
	Hence, agent $u_i$ is possibly envy-free in $I'$. 
	Now, consider some $i \in [k+1, n]$. There exists a preference list in $I'$ where $\omega'(u_i) = e_{i-k}$ is the most preferred house, and so agent $u_i$ is possibly envy-free in $I'$.
	
	For the other direction, suppose we have a possibly envy-free allocation in $I'$.
	At most $n-k$ agents were allocated houses in $\{e_1, \ldots, e_{n-k}\}$ and so at least $k$ agents were allocated houses in $\{h_1, \ldots, h_m\}$. These agents remain possibly envy-free if they are allocated the same house in $I$.  
\end{proof}

As a corollary, we get the following result.

\begin{corollary} For the lottery model, there is no polynomial-time algorithm with bounded multiplicative approximation ratio for {\sc Max-ProbEF}, assuming P$\neq$NP.
\end{corollary}

%

\section{Compact Indifference Model}
In this section, we show that, in contrast to the lottery model, {\sc ExistsPossiblyEF} and  {\sc ExistsCertainlyEF} are polynomial-time solvable for the Compact Indifference model. However, perhaps surprisingly,  multiplicative approximation to {\sc Max-ProbEF}  still remains hard.  We will then prove that for any fixed $\epsilon >0$, there exists a polynomial-time algorithm such that it either computes the optimal solution for  {\sc Max-ProbEF}  exactly or returns a certificate showing that  every allocation has probability of EF less than $\epsilon$.

\subsection{Complexity Results}


%
%
%

Our first complexity result is a strong inapproximability one. 

\begin{theorem}\label{prop:hardness-of-highestEF-under-compact-indifference}
	Let $f(n, m)$ be a polynomial in the number of agents and houses. 
	Then, finding an $f(n, m)$-approximation of {\sc Max-ProbEF} under the compact indifference model is NP-hard.
\end{theorem}
\begin{proof}	
	We begin by describing how to find the probability that a given allocation is EF.
	In particular, consider some allocation $\omega$ where agent $i$ is allocated house $\omega(i)$.
	Firstly, if $\omega(j) \succ_i \omega(i)$ for any $i, j \in N$ then the allocation is envy-free with probability 0.
	Otherwise, for each agent $i$, let $X_i = \{ j \in N : \omega(i) \sim_i \omega(j) \}$. 
	Then, the probability of agent $i$ being envy-free is $\frac{1}{|X_i|}$, and, by independence, the probability that the allocation is envy-free is $\Pi_{i \in N}   \frac{1}{|X_i|}$. 
	   
	We prove NP-hardness via a reduction from {\sc Independent Set}.
	Firstly, for convenience, we only provide partial preference lists in the proof. 
	In particular, for each agent $i$, we provide a subset of houses $H' \subseteq H$ and a weak preference list over these houses. 
	For all other houses in $H \setminus H'$, we assume that agent $i$ values these strictly worse than all houses in $H'$, and that agent $i$ cannot be allocated any house in $H \setminus H'$ in a possibly-EF allocation.
	To do this, we introduce two new agents $a_1, a_2$ and two new houses $e_1, e_2$. Agents $a_1$ and $a_2$ have the following preferences:
	\begin{itemize}
	\item $a_1$: $e_1 \succ_{a_1} e_2 \succ_{a_1}$ all houses in $H$, in some arbitrary strict ordering.
	\item $a_2$: $e_2 \succ_{a_2} e_1 \succ_{a_2}$ all houses in $H$, in the same arbitrary strict ordering.
	\end{itemize}
	Then, any possibly-EF allocation $\omega$ has $\omega(a_1) = e_1$ and $\omega(a_2) = e_2$.
	Now, reconsider agent $i$ and the subset $H' \subseteq H$ of houses that they have a preference list over. 
	We can extend this preference list into a complete preference list over all houses in $H \cup \{e_1, e_2\}$. In particular, assume that agent $i$'s preference list over $H'$ is $h_1 \succ_i \ldots \succ_i h_{|H'|}$ (note that this list could include weak preferences). Then, agent $i$'s preference list over $H \cup \{e_1, e_2\}$ is: $h_1 \succ_i \ldots \succ_i h_{|H'|} \succ_i e_1 \succ_i e_2 \succ_i$ all houses in $H \setminus H'$, in any arbitrary order. 
	Then, agent $i$ cannot be allocated any house in $H \setminus H'$ (otherwise, agent $i$ would envy agents $a_1$ and $a_2$), and so for convenience we omit these houses from the preference list.
	
	We now describe the reduction. 
	First, note that any polynomial $f(n, m)$ can be upper bounded by a function of the form $a(n+m)^r$ for some positive integers $a$ and $r$. So, we assume that $f(n, m)$ is of this form. 
	Now, consider an instance $I$ of {\sc Independent Set} with a graph $G = (V, E)$, and a target $k$. The goal is to determine if there exists an independent set in $G$ of size $k$.

	We construct an instance $I'$ of {\sc Max-ProbEF} as follows. 	
	For each vertex $v \in V$, we introduce two houses $t_v$ and $f_v$, and an agent $a_v$.
	We will design our instance so that house $t_v$ will be allocated to agent $a_v$ if $v$ is in the independent set, and $f_v$ will be allocated to agent $a_v$ otherwise. We will show later that no other agent can be allocated $t_v$ or $f_v$.
	In particular, agent $a_v$'s preference list is $t_v \sim_{a_v} f_v$. 
	
	Our goal is to find a large independent set, and so we would rather house $t_v$ be allocated.
	We do this using a \emph{single penalty gadget}, where we apply a penalty for allocating house $f_v$. 
	In a single penalty gadget, we add two new agents $a_1, a_2$ and four new houses $e_1, e_2, e_3, e_4$:
		\begin{itemize}
			\item $a_1$'s preference list: $e_1 \sim_{a_1} e_2 \succ_{a_1} f_v \succ_{a_1} e_3$.
			\item $a_2$'s preference list: $e_1 \sim_{a_2} e_2 \succ_{a_2} f_v \succ_{a_2} e_4$.
		\end{itemize}
	First, note that house $f_v$ cannot be allocated to agent $a_1$ nor $a_2$, because doing so is impossible without one of the agents being envious. 
	Now, if $f_v$ is unallocated, then $a_1$ can be allocated house $e_3$ and $a_2$ can be allocated house $e_4$, so that both agents are certainly EF. Otherwise, if $f_v$ is allocated, then agents $a_1$ and $a_2$ must be allocated houses $e_1$ and $e_2$, giving each agent a $\frac{1}{2}$ probability of EF. Hence, this gadget multiplies the probability of EF by $\frac{1}{4}$ if house $f_v$ is allocated.
	
	Let $\alpha = 49ar^2  |V| |E|$.
We create $\alpha$ copies of the single penalty gadget for each house $f_v$. Hence, if $f_v$ is allocated (and hence, vertex $v$ is not in the independent set), then the probability of EF is multiplied by $\frac{1}{4^{\alpha}}$. 
	
	Now, for each edge $uv \in E$, either vertex $u$ or $v$ must not be in the independent set. 
	Thus, at least one of $t_u$ and $t_v$ must be unallocated. 
	We use $\emph{double penalty gadgets}$ for this purpose. 
	A double penalty gadget is built for two houses $h_1, h_2$ and adds  four new agents $a_1, a_2, a_3, a_4$ and eight new houses $e_1, e_2, e_3, e_4, e_5, e_6, e_7, e_8$:
		\begin{itemize}
			\item $a_1$'s preference list: $e_1 \sim_{a_1} e_2 \sim_{a_1} e_3 \sim_{a_1} e_4 \succ_{a_1} h_1 \succ_{a_1} e_5.$
			\item $a_2$'s preference list: $e_1 \sim_{a_2} e_2 \sim_{a_2} e_3 \sim_{a_2} e_4 \succ_{a_2} h_1 \succ_{a_2} e_6.$
			\item $a_3$'s preference list: $e_1 \sim_{a_3} e_2 \sim_{a_3} e_3 \sim_{a_3} e_4 \succ_{a_3} h_2 \succ_{a_3} e_7.$
			\item $a_4$'s preference list: $e_1 \sim_{a_4} e_2 \sim_{a_4} e_3 \sim_{a_4} e_4 \succ_{a_4} h_2 \succ_{a_4} e_8.$
		\end{itemize}
	First, note that neither $h_1$ nor $h_2$ can be allocated to any of these agents. 
	In particular, without loss of generality, if $h_1$ is allocated to agents $a_1$ or $a_2$, then at least one of agent $a_1$ and $a_2$ will always be envious of the other. 
	Now, if either $h_1$ or $h_2$ is allocated to some other agent, then one of the houses $e_1, e_2, e_3, e_4$ will be allocated to $a_1, a_2, a_3$ or $a_4$. Then, to avoid envy, all the agents $a_1, a_2, a_3, a_4$ will be allocated houses $e_1, e_2, e_3, e_4$ in some permutation. 
	Therefore each agent will be EF with probability $\frac{1}{4}$, and so all the agents are EF with probability $\frac{1}{256}$. 
	However, if both $h_1$ and $h_2$ are unallocated, then agents $a_1, a_2, a_3, a_4$ can be allocated houses $e_5, e_6, e_7, e_8$ respectively, and will all be EF with probability $1$. 
	Hence, this gadget multiplies the EF-probability by $\frac{1}{256}$ if either of $h_1$ or $h_2$ are allocated.
	
	For each edge $uv \in E$, we add $|V|\alpha$ copies of this gadget for the following pairs of houses: $(t_u, t_v)$, $(t_u, f_v)$, $(f_u, t_v)$. 
	We claim that these gadgets together provide a penalty if both houses $t_u$, $t_v$ are allocated:
	\begin{itemize}
		\item If both $t_u$ and $t_v$ are allocated, then all $3|V|\alpha$ gadgets provide penalty. Hence, the EF-probability is multiplied by $\frac{1}{256^{3|V|\alpha}}$.
		\item If at most one of $t_u$ and $t_v$ are allocated, then exactly $2|V|\alpha$ of the gadgets provide penalty. In particular, if neither $t_u$ and $t_v$ are allocated, then both $f_u$ and $f_v$ will be allocated, and so the $(t_u, f_v)$ and $(f_u, t_v)$ gadgets provide penalty, but the $(t_u, t_v)$ gadgets do not. On the other hand, if $t_u$ is allocated but $t_v$ is unallocated, then $f_v$ is allocated and so the $(t_u, t_v)$ and $(t_u, f_v)$ gadgets provide penalty but the $(f_u, t_v)$ gadgets do not. Finally, the case when $t_u$ is unallocated but $t_v$ is allocated is symmetric. Thus, the EF-probability is multiplied by  $\frac{1}{256^{2|V|\alpha}}$.
	\end{itemize}
	Therefore, if both $t_u$ and $t_v$ are allocated, then the EF-probability is multiplied by an additional $\frac{1}{256^{|V|\alpha}}$, compared to the case where this does not happen.
	
	This completes the description of the instance $I'$.
	It can be shown that the {\sc Max-ProbEF} instance $I'$ is polynomial in size, and has the same answer as the {\sc Independent Set} instance $I$. This is detailed in the appendix.
\end{proof}

Next, we complement the  above result by showing that  additive approximations for {\sc Max-ProbEF} are computationally tractable.

\subsection{Algorithm for {\sc Max-ProbEF} }
Let $\mathsf{OPT}$ be EF probability of the optimal solution to {\sc Max-ProbEF}.
\begin{theorem}	\label{thm:compact-indiff-additive}
For any fixed $\epsilon >0$, there exists an algorithm running in polynomial time such that it either 
\begin{itemize}
\item computes $\mathsf{OPT}$ exactly, or
\item returns a certificate that $\mathsf{OPT}<\epsilon$ i.e., every allocation has probability of EF less than $\epsilon$ 
\end{itemize}
\end{theorem}

%

To prove the theorem, we show that there is a polynomial-time algorithm (\Cref{alg:compact-indiff-additive}) with the above properties. Next, we specify some terminology and machinery for the algorithm and the corresponding proof. 
Consider some instance $I$ under the compact indifference model, and assume that there exists an allocation $\omega$ with a non-zero EF-probability. 
Then, this allocation $\omega$ must be envy-free with respect to the underlying weak deterministic preferences, that is, $\omega(i) \succsim_i \omega(j)$ for all agents $i, j$.
We define an $n \times n$ binary matrix $A$, which we call the \emph{envy-matrix} of $\omega$, as follows:
$$
A_{i, j} =
\begin{cases}
0, & \text{if $\omega(i) \succ_i \omega(j)$} \\
1, & \text{otherwise (i.e. $\omega(i) \sim_i \omega(j)$)} \\
\end{cases}
$$

\begin{lemma}
\label{lem:compact indifference envy-matrix}
Let $\epsilon$ be a constant satisfying $0 < \epsilon \leq 1$.
Consider an instance under the compact indifference model that admits an allocation $\omega$ with EF-probability at least $\epsilon$. Let $A$ be the envy-matrix of $\omega$. Then:
\begin{enumerate}
	\item The EF-probability of $\omega$ is $\prod_{i=1}^{n} \frac{1}{\sum_{j=1}^n A_{i, j}}$, and
	\item  Excluding the main diagonal, the number of 1s in the envy-matrix is at most $\frac{1}{\epsilon}$.
\end{enumerate}
\end{lemma}
{For the first condition, a single agent is envy-free with probability $\frac{1}{\sum_{j=1}^n A_{i, j}}$, which is equivalent to the formula described in the proof of \Cref{prop:efprob}. Due to independence, we multiply these values for each agent.
The second condition can be derived from the first condition and the assumption that the EF-probability is at least $\epsilon$. This is detailed in the appendix.}

We utilise envy-matrices to prove \Cref{thm:compact-indiff-additive}. 
To this end, we introduce the {\sc AllocSatisfyingEnvyMatrix} problem. 
In this problem, each agent has weak deterministic preferences over the houses.
Additionally, an $n \times n$ matrix $A$ is supplied.
The goal is to find an allocation $\omega$ satisfying the following condition, or determine that such an allocation does not exist. For all agents $i, j$:
\begin{itemize}
	\item If $A_{i, j} = 1$, then $\omega(i) \succsim_i \omega(j)$.
	\item Otherwise, if $A_{i, j} = 0$, then $\omega(i) \succ_i \omega(j)$.
\end{itemize}
Informally, the goal is to find an allocation that is ``at least as good'' as the envy matrix. 

\begin{lemma}
	\label{lem:envyMatrixAlgo}
	{\sc AllocSatisfyingEnvyMatrix} can be solved in polynomial time.
\end{lemma}

{ We prove \Cref{lem:envyMatrixAlgo} by extending the algorithm of  \citet{GSV19}. As an overview, our algorithm either (i) finds an allocation, or (ii) identifies a set of houses, such that none of these houses can appear in any allocation satisfying the envy-matrix. In the case of (ii), these houses are deleted and the procedure repeats. This terminates when either an allocation is found, or less than $n$ houses remain, and so no allocation exists. 
This is detailed in the appendix.}

	\begin{algorithm}
        \caption{Additive Approximation for {\sc Max-ProbEF}}
        \label{alg:compact-indiff-additive}
        \small
        \textbf{Input:} A compact indifference instance $I$ and constant $\epsilon$ satisfying $0 < \epsilon \leq 1$.
        
        \textbf{Output:} An allocation with maximum EF-Prob, if there is one with EF-Prob at least $\epsilon$; otherwise, returns null.
        \begin{algorithmic}[1]
        	\STATE $allocs \gets $ empty array
        	
        	\FOR{each $n \times n$ binary matrix $A$ where $A_{i,i} = 1$ for all $i \in N$ and there are at most $\frac{1}{\epsilon}$ 1s outside the main diagonal} \label{alg:for-loop}
        		\IF{$\frac{1}{\sum_{i=1}^n \sum_{j=1}^n A_{i, j}} \geq \epsilon$}
        			\STATE Create an {\sc AllocSatisfyingEnvyMatrix} instance with envy-matrix $A$ and the same weak preferences as $I$. Solve this using \Cref{lem:envyMatrixAlgo} and let the result be $\omega$.\label{line:EMA}
        			\IF{$\omega \neq \text{null}$}
        				\STATE Append $\omega$ to $allocs$
        			\ENDIF
        		\ENDIF
        	\ENDFOR
        	
        	\IF{$allocs$ is non-empty}
        			\RETURN an allocation in $allocs$ with maximum EF-Prob
        	\ELSE
        			\RETURN no solution
        	\ENDIF
        \end{algorithmic}
    \end{algorithm}	

\begin{proof}[Proof of \Cref{thm:compact-indiff-additive}]

	We introduce \Cref{alg:compact-indiff-additive} and prove its correctness.
	Consider an instance $I$. Let $\omega$ be an allocation with maximum EF-probability in $I$, and let this EF-probability be $p$.
	We first consider the case where $p \geq \epsilon$.
	Let $A$ be the envy-matrix of $\omega$.
	From \Cref{lem:compact indifference envy-matrix}, we know that $\prod_{i=1}^{n} \frac{1}{\sum_{j=1}^n A_{i, j}} = p \geq \epsilon$ and $A$ has at most $\frac{1}{\epsilon}$ 1s (excluding the main diagonal). 
	Hence, line~\ref{line:EMA} will be run with this envy-matrix and so an allocation $\omega$ will be found with EF-probability $p$. It follows that \Cref{alg:compact-indiff-additive} finds an allocation with maximum EF-probability in this case.
	
	Now, assume that $p < \epsilon$.
	Then, every instance of {\sc AllocSatisfyingEnvyMatrix} will report no solution, and so \Cref{alg:compact-indiff-additive} will correctly determine that no allocation exists with EF-probability at least $\epsilon$.
	
	We now analyse the time complexity.
	The for loop on line~\ref{alg:for-loop} iterates over $\sum_{i = 0}^{\left \lfloor \frac{1}{\epsilon} \right \rfloor}{n^2 \choose i} = O(n^{\frac{2}{\epsilon}})$ matrices, and this can be done with only polynomial overhead.
	Since the algorithm of \Cref{lem:envyMatrixAlgo} runs in polynomial time, the overall running time is $O(n^{\frac{2}{\epsilon}} \times poly(n, m))$.
\end{proof}

Next, we remark that  {\sc ExistsCertainlyEF} can be solved in polynomial time for the compact indifference model, which is implied by \Cref{thm:compact-indiff-additive} (with $\epsilon=1$). 
\begin{corollary}
 {\sc ExistsCertainlyEF} can be solved in polynomial time for the compact indifference model.
\end{corollary}
Additionally,  {\sc ExistsPossiblyEF} can be solved in polynomial time using the algorithm of \citet{GSV19}. The proof is in the appendix.
\begin{proposition}
 \label{prop:ExistsPossiblyEF-compact-indiff}
 {\sc ExistsPossiblyEF} can be solved in polynomial time for the compact indifference model.
\end{proposition}

\section{Joint Probability Model}

We next show that even for an expansive uncertainty model such as joint probability, \textsc{ExistsCertainlyEF} is NP-complete.


\begin{theorem}\label{theorem:Joint-CertainlyEF}
For the joint probability model, \textsc{ExistsCertainlyEF} is NP-complete, even with a constant number of profiles.
\end{theorem}
\begin{proof}
\textsc{ExistsCertainlyEF} is in NP because it can be checked in polynomial time whether a given allocation is certainly EF or not.

The NP-hardness is proved via a reduction from {\sc{ExistsCertainlyEF}} under a restricted case of lottery model
where each agent has at most six linear preferences.
One can verify that the proof of Theorem \ref{theorem:Lottery-CertainlyEF} shows that {\sc{ExistsCertainlyEF}}
is NP-complete even in the restricted case of lottery model
where each agent has six preferences.

Consider an instance $I$ of the restricted lottery model, in which there are $n$ agents and $m$ houses, and moreover each agent has six linear preferences.
Denote by $P_i = \{ \succ_{i,1}, \ldots, \succ_{i,6} \}$
the set of ordinal preferences for agent $i$ under $I$. 
We now construct an instance $I'$ of the joint probability model.
Instance $I'$ has the same number $n$ of agents and $m$ of houses. 
There are six preference profiles, each with equal probability $\frac{1}{6}$. 
For $t\in [6]$, the $t$-th preference profile is $(\succ_{1, t}, \succ_{2, t}, \ldots, \succ_{n, t})$.

We now prove that a certainly EF assignment in the restricted lottery instance $I$ corresponds to a certainly EF result in the joint probability instance $I'$.
Consider any pair of agents $i$ and $j$.
Assume that agent $i$ is assigned house $h_i$ and agent $j$ is assigned house $h_j$. Then, because the allocation is envy free under $I$, we know that $h_i \succ_{i, t} h_j$ for all $t \in [6]$. 
Accordingly, in every preference profile of $I^{\prime}$,
agent $ i $ does not envy agent $ j $,
and
thus, agent $i$ certainly does not envy agent $j$ in the instance $I'$.

We now prove the other direction.
Consider any pair of agents $i$ and $j$.
Assume that agent $i$ is assigned house $h_i$ and agent $j$ is assigned house $h_j$ in some certainly EF assignment of joint probability instance $I^{\prime}$. 
Then, because the assignment is certainly EF, we know that
for every profile, agent $ i $ does not envy agent $j$,
i.e.,
$h_i \succ_{i, t} h_j$ for all $t \in [6]$. 
Hence, agent $i$ certainly does not envy agent $j$ in the restricted lottery instance $I$.

\end{proof}

Recall that the probability of an assignment $\omega$ being EF
is the summation of probability of preference profile $P_i$, under which $\omega$ is EF.
The reduction in the proof of Theorem \ref{theorem:Joint-CertainlyEF}
indeed implies the following;
(i) there exists an assignment with EF-probability one in the restricted lottery model if and only if
there exists an assignment with EF-probability one in the joint probability model;
(ii)
there does not exist an assignment with EF-probability one in the restricted lottery model if and only if
there does not exist an assignment with EF-probability strictly greater than $\frac{5}{6}$
in the joint probability model.
Then we have the following theorem.

\begin{theorem} For any constant $\epsilon>0$,
	there is no polynomial-time algorithm with a $(\frac{6}{5} - \epsilon)$-approximation ratio for
	{\sc Max-ProbEF} under the joint probability model, assuming P$\neq$NP.
\end{theorem}

In contrast, {\sc ExistsPossiblyEF} can be solved in polynomial time by running the algorithm of \citet{GSV19} on each realizable preference profile.

\begin{proposition}
\label{prop:ExistsPossiblyEF-joint}
{\sc ExistsPossiblyEF} can be solved in polynomial time for the joint probability model. 
\end{proposition}

\section{Conclusion}

In this paper, we study a fundamental problem of envy-free house allocation under uncertainty. For each of the uncertain preference models considered, we provide a complete set of complexity results. We find that, surprisingly, each model gives rise to a different set of complexity results.

Although we mainly focus on uncertainty models that assume underlying linear preferences, we also study uncertainty models that go beyond this assumption. For example, in the pairwise probability model, where each agent has pairwise independent preferences over items, we show (in the appendix) that all the problems we consider are NP-hard.

One of our central results is a polynomial-time algorithm that either solves {\sc Max-ProbEF} exactly or certifies that every allocation has a small probability of being EF in the compact indifference model. It is interesting to see if similar results can be achieved for other uncertain preference models. Another possible future direction is to find the best multiplicative approximation ratio achievable for the joint probability model. Finally, it is intriguing to investigate similar problems for general resource allocation settings. 

\section*{Acknowledgement}
This work was supported by NSF-CSIRO grant on ``Fair Sequential Collective Decision-Making". Bo Li is funded by NSFC under Grant No. 62102333 and HKSAR RGC under Grant No. PolyU 15224823. Mashbat Suzuki is partially supported by the ARC Laureate Project FL200100204 on ``Trustworthy AI". 

\bibliographystyle{aaai24}
\bibliography{abb,haris_masterEF,aziz_personalEF}

\appendix

\section{Missing Proofs}

\subsection*{Proof of Proposition~\ref{prop:efprob}}

\begin{proof}[Proof of Proposition~\ref{prop:efprob}]
We deal case by case. 
\begin{enumerate}
\item We initalize $q$ to 1.
For every preference profile $P_i$ in the support, we  check whether the given allocation $A$ is envy-free or not. If there is some agent who envious in $A$ according to preferences in $P_i$ we subtract the probability $p_i$ of profile $P_i$ from $q$. We do with all profiles in the support and then return the final value of $q$.
\item For each agent $i\in N$ we find the probability that the agent is not envious. This is done as follows. For agent $i$ we check each of its linear orders in its support and check whether $i$ is not envious of any agent. We add the probabilities of all such linear orders for agent $i$ to get the probability $q_i$ that $i$ is not envious. The probability that $A$ is envy-free is equal the probability that all agents are envy-free which is computed as $\prod_{i\in N}q_i$.
\item For each agent $i\in N$ we find the probability that the agent is not envious. This is done as follows. For agent $i$, we compute the probability $q_ij$ that is not envious of a given agent $j$. 
The probability $q_i$ that agent $i$ is not envious of any agent is equal to $\prod_{j\in N\setminus \{i\}}q_{ij}$. The probability that $A$ is envy-free is equal the probability that all agents are envy-free which is computed as $\prod_{i\in N}q_i$.
\item For each agent $i\in N$ we find the probability that the agent is not envious. This is done as follows. 
For agent $i$, let $S = \{ j \in N: A(j) \sim_i A(i) \}$. Agent $i$ will be envy-free if and only if the resulting complete linear ordering has house $A(i)$ preferred over house $A(j)$ for all $j \in S \setminus \{i \}$. Hence, agent $i$ is not envious with probability $\frac{1}{|S|}$.
To solve this in linear time, we first mark all the houses that have been allocated.
Then, for each agent $i$, we simply iterate through their weak preference order, and count the number of allocated houses that are preferred equally to house $A(i)$. 
\end{enumerate}
\end{proof}

\subsection*{Proof of \Cref{lem:ExistsPartialEF}}

\begin{proof}[Proof of \Cref{lem:ExistsPartialEF}]
        Our reduction is a modification of the reduction of \citet{KMW21} (Theorem 3.5). They prove hardness for a similar problem, where all $n$ agents must be allocated houses with the requirement that at least $k$ agents are envy-free. Most of the construction and proof is the same as \citet{KMW21}, with some changes to adapt to the different setting, such as an adjustment to the set $H$ of houses in the constructed instance.
        
        We reduce from the decision version of {\sc{Minimum Coverage}}. In this problem, there is a finite set $E$ of elements, subsets $S_1, S_2, \ldots, S_d$ of $E$, and positive integers $q, \ell$ with $q \leq |E|$ and $l \leq d$. 
        The goal is to determine whether there exists a subset $\ell \subseteq [d]$ with $|I| = \ell$ and $|\bigcup_{t \in I}S_t| \leq q$.
        
        Given an instance of {\sc{Minimum Coverage}}, we construct an instance of  {\sc{ExistsPartialEF}} with binary preferences as follows:
        \begin{itemize}
        		\item $N = \{ a_e : e \in E \} \cup \{ a_t^* : t \in [d] \}$.
		\item $H = \{ h_t^* : t \in [d]\} \cup \{h_t : t \in [|E| - q + d - \ell \}$.
		\item For each agent $i \in N$, recall that their binary preferences can be expressed in terms of two sets $A_i$ and $B_i$, such that houses in $A_i$ are preferred over houses in $B_i$. For each agent, we define $A_i$, and the set $B_i$ will be $H \setminus A_i$:
		\begin{itemize}
			\item For $e \in E$, $A_{a_e} = \{ h^*_t : e \in S_t \}$.
			\item For $t \in [d]$, $A_{a^*_t} = \{ h^*_t \}$.
		\end{itemize} 
		\item $k = |E| - q + d$.
        \end{itemize}
        
        Assume that there is a solution $I \subseteq [d]$ (with $|I| = \ell$) to the {\sc{Minimum Coverage}} instance. We show that there is a valid allocation $\omega$ for our {\sc{ExistsPartialEF}} instance. 
        For notational convenience, we use $\omega(i) = \emptyset$ to denote that agent $i$ was not allocated a house.
        \begin{itemize}
        		\item For all $t \in I$, $\omega(a_t^*) = h_t^*$.
		\item Let $E' = E \setminus \bigcup_{t \in I}S_t$. Note that $|E'| \geq |E|-q$, since $I$ is a solution to the {\sc{Minimum Coverage}} instance. Let $E'' \subseteq E'$ be an arbitrary subset with $|E''| = |E|-q$. 
		Then, we allocate the houses in $\{h_t : t \in [|E| - q + d - \ell \}$ to the agents in $\{ a_e : e \in E'' \} \cup \{ a^*_t : t \in [d] \setminus I \}$ in an arbitrary way.
		\item For every remaining agent $i$, $\omega(i) = \emptyset$.
        \end{itemize}
        First, note that exactly $|I| + (|E|-q + d - \ell) = |E| - q + d = k$ agents were allocated houses. We now show that that these agents are envy-free:
        \begin{itemize}
        		\item Consider some $t \in I$. Since $\omega(a_t^*) \in A_{a_t^*}$, we know that $a_t^*$ is not envious.
		\item Consider some $e \in E''$. By definition of $E''$, we know that all houses in $A_{a_e}$ are unallocated, and so $a_e$ is not envious.
		\item Consider some $t \in [d] \setminus I$. Since house $h^*_t$ is unallocated, agent $a_t^*$ is not envious.
        \end{itemize}
        This completes the proof of the first direction.
        
        Now, assume that there is an allocation $\omega$ for the {\sc{ExistsPartialEF}} instance. In particular, exactly $k$ agents are allocated houses and these agents are envy-free. We will show that there is a solution to the {\sc{Minimum Coverage}} instance.
        
        \begin{claim}\label{claim:ExistsPartialEF-fixalloc}
        		There exists an allocation $\omega'$ satisfying the following conditions:
		\begin{enumerate}
			\item Exactly $k$ agents are allocated houses in $\omega'$ and these agents are envy-free.
			\item For each $t \in [d]$, if $h_t^*$ is allocated then $\omega'(a_t^*) =  h_t^*$.
		\end{enumerate}
        \end{claim}
        \begin{proof}[Proof of \Cref{claim:ExistsPartialEF-fixalloc}]
        		Note that $\omega$ satisfies the first condition, and so assume that $\omega$ does not satisfy the second condition.
		Then, let $t \in [d]$ be a value for which it is not satisfied: that is, $h_t^*$ is allocated but $\omega(a_t^*) \neq h_t^*$.
		Then, $\omega(a_t^*) = \emptyset$, or otherwise agent $a_t^*$ would envy whomever is allocated $h_t^*$.
		
		Let $i \in N$ be the agent such that $\omega(i) = h_t^*$.
		We define new allocation $\omega'$, which is the same as $\omega$ except that $\omega'(a_t^*) = h_t^*$ and $\omega'(i) = \emptyset$.
		Note that $\omega'$ still satisfies the first condition, since agent $a_t^*$ is envy-free, agent $i$ is not allocated a house and every other agent is unchanged from $\omega$. 
		
		By repetitively applying this process, we obtain an allocation satisfying the second condition.
        \end{proof}
        We use \Cref{claim:ExistsPartialEF-fixalloc} to obtain an allocation $\omega'$ satisfying the two conditions.
        Then, let $I = \{ t \in [d] : h^*_t \text{ is allocated in } \omega' \}$.
        Note that $|I| \geq \ell$: otherwise, if $|I| < \ell$, then the number of allocated houses in $\omega'$ would be less than $\ell + (|E|-q+d-\ell) = k$.
        We will prove that $I$ is a valid solution to the {\sc{Minimum Coverage}} instance.

        	Let $E' = \{ e \in E : \omega'(a_e) \neq \emptyset \}$. 
	Since $k$ agents were allocated houses in $\omega'$, it follows that $|E'| \geq k-d = |E|-q$. 
	Now, consider some $e \in E'$. 
	From \Cref{claim:ExistsPartialEF-fixalloc}, we know that $\omega'(a_e) \in \{h_t : t \in [|E| - q + d - \ell \}$ and so $\omega'(a_e) \in B_{a_e}$. 
	Hence, we know that $e \not\in \bigcup_{t \in I}S_t$: 
 otherwise, agent $a_e$ would envy at least one agent from $\{ a^*_t : t \in I \}$. 
	
	Hence, $E' \subseteq (E \setminus \bigcup_{t \in I}S_t)$ and so $|E \setminus \bigcup_{t \in I}S_t| \geq |E|-q$.
	Thus, $|\bigcup_{t \in I}S_t| \leq q$, completing the proof.
\end{proof}

\subsection*{Proof of Lemma~\ref{lem:ExistsPartialPossiblyEF}}

\begin{proof}[Proof of Lemma~\ref{lem:ExistsPartialPossiblyEF}]
	We reduce from {\sc{ExistsPartialEF}} under binary preferences, which is NP-Hard from Lemma~\ref{lem:ExistsPartialEF}.
	Consider an instance $I$ of {\sc{ExistsPartialEF}} with $n$ agents and $m$ houses $\{h_1, \ldots, h_m\}$ and parameter $k$.
	
	We construct an instance $I'$ of {\sc{ExistsPartialPossiblyEF}} under the lottery model with $n$ agents, $m$ houses, and parameter $k$. 
	We now construct the preference lists for the agents. 
	Fix agent $i$. Suppose that in $I$ agent $i$ has partitioned the houses into two sets $A_i, B_i$ with houses in the same set valued equally and $h_a \succ h_b$ for all $h_a \in A_i$ and $h_b \in B_i$. 
	Further assume that $A_i = \{a_1, \ldots, a_{|A_i|}\}$ and $B_i = \{b_1, \ldots, b_{|B_i|}\}$.
	Then, in $I'$, agent $i$ has $|A_i| + |B_i|$ preference lists as follows:
	\begin{itemize}
		\item For each $j \in [|A_i|]$, there is a preference list $h_{a_j} \succ_j h_{a_1} \succ_j \ldots \succ_j h_{a_{j-1}} \succ_j h_{a_{j+1}} \succ_j \ldots \succ_j h_{a_{|A_i|}} \succ_j h_{b_1} \succ_j \ldots \succ_j h_{b_{|B_i|}}$. Note that the $h_{a_1} \succ_j \ldots \succ_j h_{a_{j-1}}$ segment is empty if $j = 1$, and the $h_{a_{j+1}} \succ_j \ldots \succ_j h_{a_{|A|}}$ segment is empty if $j = |A|$.
		\item For each $j \in [|B_i|]$, there is a preference list $h_{a_1} \succ_j \ldots \succ_j h_{a_{|A_i|}} \succ_j h_{b_j} \succ_j h_{b_1} \succ_j \ldots \succ_j h_{b_{j-1}} \succ_j h_{b_{j+1}} \succ_j \ldots \succ_j h_{b_{|B_i|}}$. Note that the $h_{b_1} \succ_j \ldots \succ_j h_{b_{j-1}}$ segment is empty if $j = 1$, and the $h_{b_{j+1}} \succ_j \ldots \succ_j h_{b_{|B_i|}}$ segment is empty if $j = |B_i|$.
	\end{itemize}
	
	This completes the description of the reduction. For correctness, consider some agent $i$. In every preference list in $I'$, all houses in the set $A_i$ are preferred over all houses in $B_i$. 
	Additionally, for each house $a \in A_i$ (resp $b \in B_i$) there exists a preference order where house $a$ (resp. $b$) is preferred over all other houses in $A_i$ (resp. $B_i$). Hence, an agent being envy-free in the instance $I$ is equivalent to an agent being possibly envy-free in the instance $I'$.
	\end{proof}

\subsection*{Proof of \Cref{prop:hardness-of-highestEF-under-compact-indifference}}

\begin{proof}[Proof of \Cref{prop:hardness-of-highestEF-under-compact-indifference}]
In the paper, we provide a reduction from an instance $I$ of {\sc Independent Set} to an instance $I'$ of  {\sc Max-ProbEF} under the compact indifference model.
Here, we show that this reduction is polynomial in size and that the answer to the instance $I$ is the same as the answer to the instance $I'$.

We first show that the instance $I'$ is polynomial in size. 
	The reduction creates the following agents and houses:
	\begin{itemize}
		\item Two agents and two houses are introduced so that we can provide partial preference lists.
		\item For each vertex $v \in V$, we introduce $1$ agent and $2$ houses. Additionally, we introduce $\alpha$ copies of the single penalty agent, each of which has $2$ agents and $4$ houses. Hence, each vertex contributes $2\alpha + 1$ agents and $4\alpha + 2$ houses.
		\item For each edge $uv \in E$, we introduce $3|V|\alpha$ double penalty gadgets, each of which has $4$ agents and $8$ houses. Hence, each edge contributes $12|V|\alpha$ agents and $24|V|\alpha$ houses.
	\end{itemize}
	Overall, there are $n =2 + |V| + 2|V|\alpha + 12|V||E|\alpha$ agents and $m = 2 + 2|V| + 4|V|\alpha + 24|V||E|\alpha$ houses.
	Since $\alpha$ is polynomial in $|V|$ and $|E|$, the instance $I'$ is polynomial in size.
	
	\begin{claim}\label{claim:possibly-EF}
		Consider some possibly-EF allocation in the instance $I'$, where houses inside of single and double penalty gadgets are allocated in such a way that maximises the probability of envy-freeness. Then, this allocation satisfies the following conditions:
		\begin{enumerate}
			\item For each vertex $v \in V$, exactly one of houses the $t_v$ and $f_v$ is allocated.
			\item Let $\ell = |\{ v \in V : t_v \text{ is allocated} \}|$, and let $o = |\{ uv \in E : \text{both houses $t_u$ and $t_v$ are allocated} \}|$. Then, the probability that the allocation is EF is $\frac{1}{256^{2|E||V|\alpha  + o|V|\alpha}4^{(|V|-\ell)\alpha}}$.
		\end{enumerate}
	\end{claim}
	\begin{proof}[Proof of Claim~\ref{claim:possibly-EF}]
		For the first condition, note that for each $v \in V$, one of $t_v$ and $f_v$ must be allocated to agent $a_v$. Additionally, we showed that no other agent can be allocated $t_v$ or $f_v$ in any possibly-EF allocation, completing the proof of this condition.
		
		For the second condition, we begin by noting that the only uncertainly of EF comes from the gadgets. 
		The single penalty gadget introduces a penalty of $\frac{1}{4^{\alpha}}$ for each house $f_v$ that is allocated.
		Therefore, these gadgets multiply the probability of EF by a combined penalty of $\frac{1}{4^{(|V|-\ell)\alpha}}$.
		The double penalty gadgets introduce a penalty of $\frac{1}{256^{2|V|\alpha}}$ for each edge $uv \in E$, with an additional penalty of $\frac{1}{256^{|V|\alpha}}$ if both $t_u$ and $t_v$ are allocated. Hence, the combined penalty of these gadgets is $\frac{1}{256^{2|E||V|\alpha}} \times \frac{1}{256^{o|V|\alpha}}$. By independence, we can multiply the penalty of the single and double penalty gadgets to get the overall probability of envy-freeness.
	\end{proof}

	It follows from Claim~\ref{claim:possibly-EF} that any possibly-EF allocation in $I'$ has EF-probability at most $\frac{1}{256^{2|E||V|\alpha}}$.
	
	We now show how we can use the instance $I'$ to find an independent set in the instance $I$.

	\begin{claim}\label{claim:approx-factor}
	The following two statements hold:
	\begin{enumerate}
		\item There exists an independent set of size $k$ in $G$ if and only if there exists an allocation in $I'$ with EF-probability $\frac{1}{256^{2|V|\alpha}4^{(|V|-k)\alpha}}$.
		\item There does not exist an independent set of size $k$ in $G$ if and only if there does not exist an allocation in $I'$ with EF-probability strictly greater than $\frac{1}{256^{2|V|\alpha}4^{(|V|-k+1)\alpha}}$.
	\end{enumerate}
	\end{claim}
	\begin{proof}[Proof of Claim~\ref{claim:approx-factor}]
	We begin by proving the first statement. 
	Assume that we have an independent set $S$ of $G$ with $|S| = k$.
	We create an allocation $\omega$ as follows: 
	For each $v \in V$, if $v \in V$ we set $\omega(a_v) = t_v$. Otherwise, we set $\omega(a_v) = f_v$.
	We allocate houses inside gadgets in such a way that maximises envy-freeness. 
	Then, by Claim~\ref{claim:possibly-EF} it holds that the probability of EF for $\omega$ is $\frac{1}{256^{2|V|\alpha}4^{(|V|-k)\alpha}}$. In particular, $\ell = k$ and $o = 0$, since $S$ is an independent set of $G$.
	We now prove the reverse direction. Assume that we have an allocation $\omega'$ with an EF probability of $\frac{1}{256^{2|V|\alpha}4^{(|V|-k)\alpha}}$.
	Let $S' = \{ v \in V : \text{house $t_v$ is allocated in $\omega'$} \}$. 
	We first prove by contradiction that $S'$ is an independent set of $G$. 
	As mentioned earlier, any allocation in $I'$ has EF probability at most $\frac{1}{256^{2|E||V|\alpha}}$. 
	Assume that $S'$ is not an independent set, and so there exists two vertices $u, v \in S'$ such that $uv \in E$.
	Thus, the double penalty gadgets multiply the EF probability by $\frac{1}{256^{|V|\alpha}}$, implying that the EF probability of $\omega'$ is at most $\frac{1}{256^{2|E||V|\alpha + |V|\alpha}}$. However, since $k \geq 0$, it follows that $\frac{1}{256^{2|E||V|\alpha + |V|\alpha}} < \frac{1}{256^{2|V|\alpha}4^{(|V|-k)\alpha}}$ and so by contradiction $\omega'$ must be an independent set. 
	We now prove by contradiction that $|S'| \geq k$.  
	Assume $|S'| < k$, and so $|V \setminus S'| > |V|-k$. 
	Then, there will be (at least) $(|V|-k+1)\alpha$ single penalty gadgets that multiply the EF probability by a combined $\frac{1}{4^{(|V|-k+1)\alpha}}$, and so $\omega'$ would have an EF probability of at most $\frac{1}{256^{2|V|\alpha}4^{(|V|-k+1)\alpha}}$, which is a contradiction.
	Hence $S'$ is an independent set of size at least $k$, completing the proof of the first statement.
	
	We now prove the second statement. 
	First, assume that all independent sets of $G$ have size strictly less than $k$. 
	Then, consider any possibly-EF allocation $\omega$.
	We denote 
	 $$S =  \{ v \in V : \text{house $t_v$ is allocated in $\omega$} \}$$ 
	We have two cases:
	\begin{enumerate}
		\item $S$ is not an independent set. Then, using the same ideas as in the previous paragraph, we know that the EF-probability of $\omega$ is at most $\frac{1}{256^{2|E||V|\alpha+ |V|\alpha}} < \frac{1}{256^{2|V|\alpha}4^{(|V|-k+1)\alpha}}$.
		\item $S$ is an independent set. Then, $|S| < k$ and so $|V \setminus S| > |V|-k$. Using the same ideas as in the previous paragraph we know that the EF probability $\omega$ is at most $\frac{1}{256^{2|V|\alpha}4^{(|V|-k+1)\alpha}}$.
	\end{enumerate}
	Now, for the other direction, assume that there does not exist an allocation in $I'$ with EF-probability strictly greater than $\frac{1}{256^{2|V|\alpha}4^{(|V|-k+1)\alpha}}$. Then, from the first statement we know that there does not exist an independent set of size $k$. 
	This completes the proof of the second statement.
	\end{proof}
	
	Note that  $$\frac{\frac{1}{256^{2|V|\alpha}4^{(|V|-k)\alpha}}}{\frac{1}{256^{2|V|\alpha}4^{(|V|-k+1)\alpha}}} = \frac{256^{2|V|\alpha}4^{(|V|-k+1)\alpha}}{256^{2|V|\alpha}4^{(|V|-k)\alpha}} = 4^\alpha$$
	Hence, by Claim~\ref{claim:approx-factor}, we know that, for any $\epsilon > 0$, a $(4^{\alpha}-\epsilon)$-approximation for the instance $I'$ is sufficient to solve the instance $I$ of {\sc Independent Set}.
	We will show that $f(n, m) < 4^\alpha$, and hence it will follow that $f(n, m)$-approximating {\sc Max-ProbEF} is NP-Hard.
	
	Recall that $f(n, m)$ is upper bounded by $a(n+m)^r$, that $n = 2 + |V| + 2|V|\alpha + 12|V||E|\alpha$, and that $m = 2 + 2|V| + 4|V|\alpha + 24|V||E|\alpha$.
	We assume that $|V| >0$ and $|E| > 0$ (otherwise, it is trivial to solve {\sc Independent Set}). Then, $n+m = 4+3|V|+6|V|\alpha+ 36|V||E|\alpha \leq 49|V||E|\alpha$. 
	It follows that 
	\begin{align*} 
		f(n, m) \leq a(n+m)^r \leq (a(n+m))^r \\ 
		\leq(49a|V||E|\alpha)^r = ((49ar|V||E|)^2)^r,
	\end{align*}
	 where the last transition is true because $\alpha = 49ar^2  |V| |E|$.
	However, note that for all $x \geq 1$ it holds that $x^2 < 4^x$. Hence, we know that $(49ar|V||E|)^2 < 4^{49ar|V||E|}$ and so
	$$f(n, m) < (4^{49ar|V||E|})^r = 4^{49ar^2|V||E|} = 4^{\alpha},$$
	completing the proof.
\end{proof}


\subsection*{Proof of \Cref{lem:compact indifference envy-matrix}}

\begin{proof}[Proof of \Cref{lem:compact indifference envy-matrix}]
	We begin with (i).
	Consider some agent $i$, and let $x_i = \sum_{j=1}^n A_{i, j}$.
	Specifically, $x_i$ is the number of agents $j$ for which $\omega(i) \sim_i \omega(j)$. 
	Recall that the probability that agent $i$ is EF under the allocation $\omega$ is $\frac{1}{x_i}$. 
	Since the EF-probability of each agent is independent, we can multiply these probabilities to obtain the given equation. 
	
	We now prove (ii). First, note that $A_{i, i} = 1$ for all $i \in N$, and so $x_i \geq 1$ for all $i$. 
	If the envy-matrix has no 1s outside the main diagonal, then (ii) trivially holds.
	Otherwise, assume that the envy-matrix has at least one 1 outside the main diagonal, and so $x_i \geq 2$ for some $i \in N$.
	Then, the EF-Probability can be re-written as 
	$$\prod_{i \in N : x_i \geq 2} \frac{1}{x_i}.$$
	
	Now, assume that there exists two agents $i, j \in N$ such that $x_i \geq x_j \geq 2$. 
	Then, $\frac{1}{x_ix_j} \leq \frac{1}{2x_i} \leq \frac{1}{x_i + x_j}$.
	By repetitively applying this inequality, we see that
	$$\prod_{i \in N : x_i \geq 2} \frac{1}{x_i} \leq \frac{1}{\sum_{i \in N : x_i \geq 2} x_i}.$$
	By assumption, the EF-Probability is at least $\epsilon$, and so 
	$$\epsilon \leq \frac{1}{\sum_{i \in N : x_i \geq 2} x_i},$$
	meaning that $\sum_{i \in N : x_i \geq 2} x_i \leq \frac{1}{\epsilon}$.
	It follows that $A$ has at most $\frac{1}{\epsilon}$  1s outside of the main diagonal.
\end{proof}

\subsection*{Proof of \Cref{lem:envyMatrixAlgo}}

\begin{proof}[Proof of \Cref{lem:envyMatrixAlgo}]
	We introduce \Cref{alg:envyMatrixAlgo}, which is a modification of the algorithm of \citet{GSV19}.
	\begin{algorithm}
	\caption{{\sc AllocSatisfyingEnvyMatrix}}
	\label{alg:envyMatrixAlgo}
	
	\textbf{Input:} A house allocation instance with weak deterministic ordinal preferences, and an $n \times n$ binary matrix $A$.
	
	\textbf{Output:} An allocation satisfying the envy matrix, or reports that there is no solution.
	
	\begin{algorithmic}[1]
		\STATE $H' \gets H$
		\WHILE{$|N| \leq |H'|$}
			\STATE We construct a bipartite graph $G = (N, H', E)$ where we create the edge set $E$ as follows. For each agent $i \in N$ and house $h \in H'$ we include the edge $(i, h)$ if $h$ is a most preferred house for agent $i$, and there does not exist another agent $j \neq i$ where $A_{j, i} = 0$ and $h$ is a most preferred house for agent $j$.
			
			\IF{there exists an $N$-saturating matching in $G$}
				\RETURN the corresponding allocation.
			\ELSE
				\STATE Find a minimal Hall violator $Z \subseteq N$, using Lemma 2.1 of \citet{GSV19}.
				
				\STATE $P \gets \{ h \in H' : \exists i \in Z$ such that $h$ is a most preferred house for agent $i \}$
				
				\STATE $H' \gets H' \setminus P$ \label{line:remove houses}
			\ENDIF
		\ENDWHILE
		
		\RETURN no solution
	\end{algorithmic}
\end{algorithm}
	First, we note that \Cref{alg:envyMatrixAlgo} is equivalent to the algorithm of \citet{GSV19} if the matrix $A$ entirely consists of 1s, but differs otherwise.
	
	\Cref{alg:envyMatrixAlgo} runs in polynomial time since a minimal Hall violator can be found in polynomial time \citep{GSV19} and each loop iteration either returns or decreases $|H'|$ by at least one.
	
	If \Cref{alg:envyMatrixAlgo} returns an allocation $\omega$, then each agent is allocated one of their most preferred houses from $H'$, and thus $\omega(i) \succsim_i \omega(j)$ for all agents $i, j$. 
	Additionally, if $A_{i, j} = 0$, then agent $j$ will not receive a house that is most preferred by agent $i$, and so $\omega(i) \succ \omega(j)$ in this case. 
	Thus, if \Cref{alg:envyMatrixAlgo} returns, then the returned allocation satisfies the envy-matrix.
	
	We now prove that the houses removed on line~\ref{line:remove houses} cannot be part of any allocation satisfying the envy-matrix. 
	Then, if \Cref{alg:envyMatrixAlgo} returns no solution, we know that no allocation exists which satisfies the envy-matrix.   
	We prove by induction on the number of iterations. 
	Consider some iteration when we have a house set $H'$.
	By the induction hypothesis, we know that any allocation satisfying the envy-matrix can only contain houses in $H'$.
	
	We introduce some notation for convenience: for a set of agents $X \subseteq N$, we define $N_G(X)$ to be neighbourhood of $X$ in $G$, and $P(X)$ to be the set of houses that are among the most preferred houses for at least one agent in $X$. Note that $N_G(X) \subseteq P(X)$ by the construction of $G$.
	
	Let $Z \subseteq N$ be a minimal hall violator of $G$. 
	Consider some agent $i \in Z$, and some house $h \in P(\{i\}) \setminus N_G(\{i\})$. 
	From the construction of $G$, there must exist some agent $j$ where $A_{j, i} = 0$ and $h \in P(\{j\})$. 
	Thus, agent $i$ cannot be allocated house $h$ in any assignment that satisfies the envy-matrix.
	 
	We now prove that the houses in $N_G(Z)$ cannot be in any allocation satisfying the envy-matrix. 
	If $N_G(Z)$ is empty then this is vacuously true, and so we assume $|N_G(Z)| > 0$. Since $Z$ is minimal, it follows that every agent $i \in Z$ has positive degree in $G$, since otherwise the set $\{i \} \subset Z$ would be a smaller Hall violator. We now proceed using the same proof as \citet{GSV19}.
	
	In particular, assume for contradiction that there exists an allocation satisfying the envy-matrix where a non-empty subset $Y' \subseteq N_G(Z)$ of houses are allocated. 
	Let $X'$ be the set of agents in $Z$ who only have edges to houses in $N_G(Z) \setminus Y'$. Since $X' \subset Z$ and $Z$ is a minimal Hall violator, it follows that $|X'| \leq N_G(X') \leq N_G(Z) \setminus Y'$. 
	Since $|Z| > |N_G(Z)|$, it follows that $|Z \setminus X'| > |Y'|$.
	From the definition of $X'$, every agent in $Z \setminus X'$ has at least one most preferred house in $Y'$. Since the houses in $N_G(Z) \setminus  Y'$ are unallocated, such an agent must be allocated a house from $Y'$ in any allocation satisfying the envy-matrix. However, there are fewer houses in $Y'$ than agents in $Z \setminus X'$, which is a contradiction.
	
	Hence, every house in $N_G(Z)$ cannot appear in any allocation satisfying the envy-matrix. 
	Since no agents in $Z$ can be allocated houses in $P(Z) \setminus  N_G(Z)$, it follows that these houses cannot be allocated either. 
\end{proof}

\subsection*{Proof of \Cref{prop:ExistsPossiblyEF-compact-indiff}}
\begin{proof}[Proof of \Cref{prop:ExistsPossiblyEF-compact-indiff}]
We do not consider a probabilistic preference input but consider the underlying weak order as deterministic. We then use the algorithm of \citet{GSV19}  to check whether an envy-free allocation exists for the weak order. Such an allocation is also envy-free for the agents where the ties are broken in a way so that agent prioritizes the item that she gets. 
\end{proof}

\section{Pairwise Model}

In the pairwise probability model, each agent reports independent pairwise probabilities over pairs of items. If agent~$i$ prefers item $o$ over item $o'$ with probability $p$, then she prefers $o'$ over $o$ with probability $1-p$. For the pairwise model, both  \textsc{ExistsCertainlyEF} and \textsc{ExistsPossiblyEF} are NP-complete.

\begin{theorem}\label{prop:PW-eef}
For the Pairwise Probability model, \textsc{ExistsCertainlyEF} is NP-complete.
\end{theorem}
\begin{proof}
\textsc{ExistsCertainlyEF} is in NP because it can be checked in polynomial time whether a given allocation is certainly EF or not.

To prove NP-hardness, we reduce from \textsc{ExistsCertainlyEF} in the lottery model, which is NP-Complete by Theorem~\ref{theorem:Lottery-CertainlyEF}. 
Consider an instance $I$ with the lottery model, with $n$ agents and $m$ houses.
We construct an instance $I'$ under the pairwise model, with the same number of agents and houses. 
Consider some agent $i$, and two houses $h$ and $h'$. 
We construct the pairwise preferences as follows:
\begin{itemize}
    \item If house $h$ occurs prior to house $h'$ in every preference list of agent $i$ under the lottery model, then $h \succ h'$ with probability $1$ for agent $i$.
    \item If house $h$ occurs after house $h'$ in every preference list of agent $i$ under the lottery model, then $h \succ h'$ with probability $0$ for agent $i$.
    \item Otherwise, house $h \succ h'$ with probability $0.5$ for agent $i$ ($0.5$ here is arbitrary, it just needs to be greater than 0 and less than 1).
\end{itemize}

We now prove that a ``yes'' result in the pairwise instance corresponds to a ``yes'' result in the lottery instance. In particular, the same allocation of houses to agents is certainly envy-free under the lottery model.

Consider any pair of agents $i$ and $j$.
Assume that agent $i$ is assigned house $h_i$ and agent $j$ is assigned house $h_j$. Then, by assumption, we know that $h_i \succ h_j$ with probability 1 for agent $i$, since the allocation is envy free under the pairwise model.
Therefore, house $h_i$ occurs prior to house $h_j$ in every preference list of agent $i$ in the lottery model.
Hence the house allocation is envy free under the lottery model.

We now prove that a ``yes'' result in the lottery instance corresponds to a ``yes'' result in the pairwise instance.

Consider any pair of agents $i$ and $j$.
Assume that agent $i$ is assigned house $h_i$ and agent $j$ is assigned house $h_j$. Then, by assumption, we know that $h_i$ must occur prior to $h_j$ in every preference list for agent $i$.
Then, in the pairwise instance, $h_i \succ h_j$ with probability $1$ for agent $i$.
Hence the house allocation is envy free under the pairwise model.
\end{proof}

To prove NP-hardness, we can reduce from \textsc{ExistsCertainlyEF} in the lottery model, which is NP-Complete by Theorem~\ref{theorem:Lottery-CertainlyEF}.

\begin{theorem}
For the Pairwise Probability model, \textsc{ExistsPossiblyEF} is NP-complete, even when agents have identical preferences. 
\end{theorem}
\begin{proof}
\textsc{ExistsPossiblyEF} is in NP because it can be checked in polynomial time whether a given allocation is possibly EF or not.

To prove NP-hardness, we reduce from \textsc{Independent Set}.
Consider an instance of independent set with a graph $G = (V, E)$ and a target $k$.
The goal is to determine whether $G$ has an independent set of size $k$.
We assume that $V = \{v_1, v_2, \ldots, v_{|V|}\}$, that is, the vertices are labelled in some arbitrary order. 

We reduce to \textsc{ExistsPossiblyEF} with identical preferences as follows. 
We create an instance of \textsc{ExistsPossiblyEF} with $n = k$ and $m = |V|$ (that is, $k$ agents and $|V|$ houses).
Intuitively, we will have one house for each vertex, and a house will be assigned to an agent if this vertex is in the independent set. 

In particular, consider two vertices $v_i$ and $v_j$ (with $i < j$), which correspond to the two houses $h_i$ and $h_j$. 
We create the pairwise preferences for all agents as follows:
\begin{itemize}
    \item If there is an edge between vertices $v_i$ and $v_j$ in $G$, then $h_i \succ h_j$ with probability $1$ (recall $i < j$, so this is well-defined for all pairs). 
    \item Otherwise, $h_i \succ h_j$ with probability $\frac{1}{2}$.
\end{itemize}

We now prove that a ``yes'' result for independent set corresponds to a ``yes'' result in the pairwise instance.
Consider an independent set $S = \{ v_{s_1}, \ldots, v_{s_k}\}$.
Then, we assign house $h_{s_i}$ to agent $i$ for all $1 \leq i \leq k$.
To prove this is possibly envy free, consider any two agents $i$ and $j$.
Because $v_{s_i}$ and $v_{s_j}$ are both in the independent set, there is no edge between $v_{s_i}$ and $v_{s_j}$ in $G$.
Hence, both agents have $h_{s_i} \succ h_{s_j}$ with probability $0.5$.
Hence, the allocation is possibly envy free. 

We now prove that a ``yes'' result for the \textsc{ExistsPossiblyEF} instance corresponds to a ``yes'' result in independent set.
Assume that agent $i$ is assigned house $h_{s_i}$, and that this allocation is envy free.
Then, let $S = \{ v_{s_1}, \ldots, v_{s_k} \}$.
We will prove that $S$ is an independent set.
Consider two vertices $v_{s_i}$ and $v_{s_j}$ in $G$, and assume without loss of generality that $s_i < s_j$.
Then, agent $i$ was assigned house $h_{s_i}$ and agent $j$ was assigned house $h_{s_j}$.
Assume for contradiction that there is an edge between $v_{s_i}$ and $v_{s_j}$ in $G$.
Then, $h_i \succ h_j$ with probability $1$, and so agent $j$ will envy agent $i$.
Thus, we have a contradiction, and so there is no edge between vertices $v_{s_i}$ and $v_{s_j}$ in $G$. 
Hence, $S$ is an independent set.
\end{proof}

\begin{corollary} For the pairwise model, there is no polynomial-time algorithm with bounded multiplicative approximation ratio for {\sc Max-ProbEF}, assuming P$\neq$NP.
\end{corollary}

\end{document}